\documentclass{article}
\usepackage{amsfonts}
\usepackage{xcolor}
\usepackage{graphicx}
\usepackage{subfig}
\usepackage{multirow}
\usepackage{color}
\usepackage{enumerate}
\usepackage{amsmath}

\newtheorem{theorem}{Theorem}

\newtheorem{corollary}[theorem]{Corollary}

\newtheorem{definition}[theorem]{Definition}

\newtheorem{proposition}[theorem]{Proposition}
\newtheorem{remark}[theorem]{Remark}

\newenvironment{proof}[1][Proof]{\noindent\textbf{#1.} }{\ \rule{0.5em}{0.5em}}
\usepackage{setspace}
\begin{document}

\title{On the optimal choice of strike
conventions in exchange option pricing}
\author{Elisa Al\`{o}s\thanks{%
Supported by grants ECO2014-59885-P and MTM2013-40782-P} \\
%EndAName
Dpt. d'Economia i Empresa \\
Universitat Pompeu Fabra \\
and Barcelona GSE\\
c/Ramon Trias Fargas, 25-27\\
08005 Barcelona, Spain\\
Email: elisa.alos@upf.edu\\
 \and 
Michael Coulon\\
Department of Business and Management\\
University of Sussex\\
Brighton BN1 9SL, UK\\
Email: m.coulon@sussex.ac.uk}
\maketitle

\begin{abstract}
An important but rarely-addressed option pricing question is how to choose appropriate strikes for implied volatility inputs when pricing more exotic multi-asset derivatives. By means of Malliavin Calculus we construct an optimal log-linear strike
convention for exchange options under stochastic volatility models. This
novel approach allows us to minimize the difference between the corresponding Margrabe computed price and the true option price.
We show that this optimal convention does not depend on the
specific stochastic volatility model chosen.  Numerical examples are given which provide strong support to the new methodology.
\end{abstract}

\vspace{2mm}


\vspace{2mm}


\section{Introduction}

Spread options are recognized as important contracts in many financial
markets, and have been widely studied both by practitioners and academic
researchers. In particular, although also traded in other markets, spread options on commodities are closely linked to
the physical markets and the hedging or valuation needs of producers and
consumers, due to their parallels with physical assets like power plants,
refineries, storage facilities or pipelines. Such assets all have an
option-like nature with operational decisions and corresponding payoffs
depending predominantly on the spread between two commodity spot or forward
prices. While a variety of different considerations affect different spread
option types (ranging from calendar spreads to locational spreads to
input/output spreads like crack or spark), the dominant derivative pricing
challenges remain the same. \newline

In particular, the commonly-used lognormal assumption (e.g. the Geometric
Brownian Motion model) for underlying prices $S_{t}^{X}$ and $S_{t}^{Y}$
leads to a convenient closed-form pricing formula known as Margrabe's
formula (see Margrabe (1978)) given an `exchange option' payoff $(S_{X}^{1}-
S_{t}^{Y})^{+}$.  In the context of stochastic volatility models, we do not have an explicit closed-form expression for the corresponding option price. Some approximations can be found for example in Demspter and Hong (2000), Antonelli, Ramponi and Scarlatti (2009), Borovkova, Permana and van der Weide (2007), Al\`{o}s and Le\'{o}n (2016) or Al\`{o}s and Rheinl\"{a}nder (2016)). All of these approaches require the previous calibration of the corresponding model parameters. In some cases, prices can only be found by simulation or other numerical methods. Computation time can be particularly
onerous for physical asset valuation or hedging, whereby strings of hourly
or daily spread options over many years or even decades are required. For such reasons, Margrabe's formula is frequently employed
for useful and fast benchmark approximations to spread option prices.
 \newline

Despite the prominence of such tools, relatively little attention has been
paid to the key question of how to choose an appropriate pair of constant
volatility inputs $\sigma_{X}$ and $\sigma_{Y}$ for Margrabe formula,
 ideally maintaining consistency both with market data and modeling
preferences. A natural starting point is the implied volatility of the two
legs of the spread, typically observable from more liquidly traded single
asset vanilla calls or puts. However, a significant implied volatility skew
or smile (as well as term structure) exists in most markets, meaning that
there are many possible choices for both $\sigma_{X}$ and $\sigma_{Y}$ and
no obvious rule for which pair is most appropriate. Indeed, there is also no
standard yardstick for measuring which so-called `strike convention' rule is
best in this setting. In Swindle (2014), this important issue is highlighted
and discussed, along with some numerical examples which indicate that the
common industry solution (described as a `volatility look-up heuristic') can
lead to significant pricing differences compared to Monte Carlo values in a
simple jump diffusion model. \newline

In this paper, we aim to answer this crucial question by developing
a new theory for an optimal short-time strike convention, defined as the choice
of implied volatilities such that the resulting estimated option price (obtained from Margrabe's formula) matches the true option price
as closely as possible. This is equivalent to the choice such that the corresponding implied correlation (backed out from Margrabe's formula) matches the model 
correlation $\rho$. It is interesting to note that both Swindle (2014) and 
Alexander and Venkatramanan (2011) comment on how the choice of strike convention can impact the implied
correlation skew, smile or frown observed across different moneyness spread
options. As the underlying assets' returns correlation is clearly unrelated
to contract moneyness, Swindle (2014) describes this as \emph{``purely an
artifact of the interaction of skew with the Margrabe formulation"},
explaining that \emph{``skew risk can manifest itself as spurious
correlation risk simply due to the look-up heuristic"}.\newline

In order to investigate such effects and recommend a strike convention for
consistent spread option pricing, we rely on tools from Malliavin calculus that allow us
to derive the short-time limit of the sensitivity of implied volatilities to
moneyness, in the context of stochastic volatility models. Our proposed optimal strike convention is, to our knowledge, the first systematic approach to this problem. Moreover, it is model-independent since it depends only on the at-the-money 
implied volatilty levels and skews of the corresponding vanilla options.   Thus, it can serve as a very useful and practical `financial engineering' tool to improve option pricing accuracy within the financial industry.\newline

The paper is organized as follow. Section 2 is devoted to introducing the main problem and notations. In Section 3 we make use of Malliavin calculus techniques to derive an equation for our strike convention proposal. In Section 4 we determine explicitly this optimal convention in the class of log-linear strike conventions. Section 5
provides a range of numerical examples and tests to investigate the theory
presented in the paper and its implications in practice. 

\section{The objective, the price model and notation}

Assume, for the sake of simplicity, that the interest rate $r=0.$ Consider a two-asset stochastic volatility model of the form
\begin{eqnarray}
&&\frac{dS^X_{t}}{S_{t}^X}=\sigma^X_{t}dW^X_{t}  \nonumber
\\
&&\frac{dS^Y_{t}}{S_{t}^Y}=\sigma^Y_{t}dW^Y_{t},
\label{themodel}
\end{eqnarray}%
under a risk-neutral probability $P$. $W^X,W^Y$ are Brownian motions and $\sigma^X_{t},\sigma^Y_{t}$
are non-negative, right-continuous and square integrable processes adapted to the filtration
generated by another Brownian motion $Z$. We will use the notation 
\[
\left\langle W_{t}^X,Z\right\rangle =\rho _X,\left\langle \;
W_{t}^Y,Z\right\rangle =\rho _Y,\left\langle \;
W_{t}^X,W_{t}^Y\right\rangle =\rho .
\]%
%Notice that there exist two independent Brownian motions $W,B$ that are also
%independent of $Z$ such that 
%\[
%W_{t}^{1}=\rho _{1}Z_{t}+\sqrt{1-\rho _{1}^{2}}W_{t},
%\]%
%and 
%\[
%W_{t}^{2}=\frac{\rho -\rho _{1}\rho _{2}}{\sqrt{1-\rho _{1}^{2}}}W_{t}+\sqrt{%
%1-\frac{\rho _{23}^{2}+\rho ^{2}-2\rho \rho _{1}\rho _{2}}{1-\rho _{1}^{2}}}%
%B_{t}
%\]

\noindent  It\^o's representation theorem gives us that, for any fixed $s$
\[
\sigma _{s}^i=E\left( \sigma _{s}^i\right) +\int_{0}^{s}a^i(s,u)dZ_{u}, \quad i=X,Y.
\]%
for some square integrable processes $a^i(s,\cdot )$ adapted to the filtration generated by\ $Z$.

Now we describe some basic notation that is used in this article. For this,
we assume that the reader is familiar with the elementary results of the
Malliavin calculus, as given for instance in Nualart (2006).

The set $\mathbb{D}_{Z}^{1,2}$ will denote the domain of the derivative operator $%
D$ with respect to the Brownian Motion $Z$. It is well-known that $\mathbb{D}_{Z}^{1,2}$ is a dense subset of $%
L^{2}(\Omega)$ and that $D$ is a closed and unbounded operator from $%
L^{2}(\Omega)$ into $L^{2}([0,T]\times\Omega).$ We will also consider the
iterated derivatives $D^{n},$ for $n>1,$ whose domains will be denoted by 
$\mathbb{D}_{Z}^{n,2}.$  We will also make use of the notation $\mathbb{L}^{n,2}:=L^{2}([0,T];\mathbb{D}_{Z}^{n,2}).$

%The adjoint of the derivative operator $D^{Z}$, denoted by $\delta^{Z},$ is
%an extension of the It\^{o} integral in the sense that the set $%
%L_{a}^{2}([0,T]\times\Omega)$ of square integrable and adapted processes
%(with respect to the filtration generated by $U$) is included in Dom$%
%\delta^{Z}$ and the operator $\delta^{Z}$ restricted to $L_{a}^{2}([0,T]%
%\times\Omega)$ coincides with the It\^{o} integral. We will make use of the
%notation $\delta^{Z}(u)=\int_{0}^{T}u_{t}dZ_{t}.$ We recall that $\mathbb{L}%
%^{n,2}:=L^{2}([0,T];\mathbb{D}_{Z}^{n,2})$ is included in the domain of $%
%\delta^{Z}$ for all $n\geq1.$

\vspace{0.2cm}

We notice that, if $\sigma^2 \in $ $\mathbb{L}^{1,2}$ the
Clark-Ocone formula gives us that%
\[
a^i(s,u)=E_{u}\left( D_{u}(\sigma^i_{s})^{2}\right), \quad i=X,Y.
\]
Then, under suitable integrability conditions the change rule for the Malliavin derivative operator (see for example Nualart (2006)) gives us that
\[
a^i(s,u)=2E_{u}\left( \sigma^i_{s}D_{u}\sigma^i _{s}\right), \quad i=X,Y.
\]

We will also make use of the following notation:

\begin{itemize}
\item $BS\left( t,x,k,\sigma\right) $ denotes the classical Black-Scholes
call price with time to maturity $T-t,$ log stock price $x$, log
strike price $k$ and volatility $\sigma$.

\item $\mathcal{L}_{BS}=\partial_t+\frac{1}{2} \sigma^2\left(\partial^2_xx-\partial_x\right)$ denotes the classical Black-Scholes operator. Notice that $(\mathcal{L}_{BS} BS)(t,x,k,\sigma)=0$.

\item $X_{t}:=\log S^X_{t}, Y_{t}:=\log S^Y_{t}$.

\item $V_{t}=E_{t}(S_{T}^X-S_{T}^Y)^+$ is the exchange option price under the model (\ref{themodel}).

\item For every $0<t<T$ and $x,k>0,I_X(t,x,z)$ is the implied volatility of an
option with payoff $(S_{T}^X-\exp \left( z\right) )^{+}$  with $X_{t}=x.$ That is, 
\[
BS\left( t,x,k,I_X(t,x,z)\right) =E_{t}(S_{T}^X-\exp \left(z \right))^{+}.
\]%
Analogously, $I_Y(t,y,z)$ is the implied volatility of an option with
payoff $\left( S_{T}^Y-\exp \left( z\right) \right) ^{+}$ with $Y_t=y.$ 

\item $\tilde v_{t}:=\sqrt{\frac{1}{T-t}\left(\int_{t}^{T}\tilde\sigma_{s}^{2}ds\right) }$

\item $M^i_t:= E_t \int_0^T (\sigma^i_s)^2 ds, \quad i=X,Y.$

\item $\tilde\sigma_t :=\sqrt{(\sigma^X_t)^{2}+(\sigma^Y_t)^{2}-2\rho \sigma^X\sigma^Y}$

\item $\tilde M_t:=E_t \int_0^T (\tilde\sigma_s)^2 ds$

\end{itemize}
For the sake of simplicity, we will take $t=0$ and we will denote $I_X(x,z)=I_X(0,x,z)$ and $I_Y(y,z)=I_Y(0,y,z)$. Moreover, we denote $x=X_0$ and $y=Y_0$.

\vspace{2mm}

It is well known that, under the Black-Scholes model, $\sigma^X_t=\sigma_X$ and $\sigma^Y_t=\sigma_Y$, for all $t\in [0,T]$ and for some positive constants $\sigma_X$ and $\sigma_Y$. In this case, the option price $V_0$ can be computed analytically by means of Margrabe's formula.  More precisely, in this case, the price is given by
\begin{eqnarray}
%\label{margrabe}
%&&V_0=BS(0,x,y,\tilde\sigma)\nonumber \\
&&BS\left(0,x,y,\sqrt{\sigma_X^2+\sigma_Y^2-2\rho\sigma_X\sigma_Y}\right)
\end{eqnarray}
In the general stochastic volatility case, there is no analytical formula for this option price. One common strategy is to substitute $\sigma_X$ and $\sigma_Y$ by the vanilla implied volatilities $I_X(x,k_X)$ and $I_Y(y,k_Y)$, for some log strikes $k_X, k_Y$. But notice that, as these implied volatilities are not constant as a function of the strike, the corresponding price estimation

\begin{equation}
\label{min2}
BS\left(0,x,y,\sqrt{ I_X^2(x,k_X)+ I_Y^2(y,k_Y)-2\rho I_X(x,k_X)I_Y(y,k_Y)}\right)
\end{equation}
will depend strongly on the choice of the log strikes $k_X$ and $k_Y$. 
Despite of the relevance of this problem, there is currently no standard rule for choosing $k_X$ and $k_Y$ (see for example Swindle (2014)). Our aim in this paper is to develop a standard rule that will allow us to choose these strikes in such a way that the approximation ($\ref{min2}$) will be as close as possible to the true option price $V_0$ for a range of moneyness cases. More precisely, we want to find the pair $k_X:=k_X(x,y)$ and $k_Y:=k_Y(x,y)$ that minimizes the difference 

\begin{equation}
\label{min}
|V_0-BS\left(0,x,y,\gamma (x,y)\right)|,
\end{equation}
for short-time and near-the-money ($x\approx y$) options, where
\begin{equation}
\label{gamma2}
\gamma(x,y):=\sqrt{ I_X^2(x,k_X)+ I_Y^2(y,k_Y)-2\rho I_X(x,k_X)I_Y(y,k_Y)}.
\end{equation}
Notice that, if we define $\hat\gamma \left(x,y\right)$  as the quantity
such that%
\begin{equation}
\label{hatgamma}
V_{0}=BS(0,x,y,\hat\gamma \left( x,y\right) ),
\end{equation}
to minimize (\ref{min}) it is sufficient to minimize
$$
{\hat\gamma}(x,y)-\gamma (x,y).
$$
\begin{remark}
Note that it is also sufficient to minimize the quantity $\rho-\hat\rho$, where $\hat\rho$ denotes the implied correlation, defined by the equality 
$$
{\hat\gamma}(x,y):=\sqrt{ I_X^2(x,k_X)+ I_Y^2(y,k_Y)-2\hat\rho I_X(x,k_X)I_Y(y,k_Y)}.
$$
\label{impliedcorrelationremark}
\end{remark}
In the following section we will develop a methodology to choose the pair $(k_X,k_Y)$. As we have no explicit expressions for $\gamma$ and $\hat\gamma$, the main idea is to approximate these two quantities and to find the pair $(k_X,k_Y)$ that makes these approximations equal.  Towards this end, we will consider for any fixed $x$ the short-time limit of the Taylor expansion of the function $\gamma(x,\cdot)-\hat\gamma(x,\cdot)$.  This motivates the following definition of strike conventions of any order.  

\begin{definition}
Assume the model  (\ref{themodel}). We will say that a pair $(k_1,k_2)\in L^2(\mathbb{R}^2;\mathbb{R}^2)$ is a short-time optimal strike convention of order $n$ (a $n$-STOSC) if
\begin{equation}
\label{STOSC}
\lim_{T\to 0} \frac{\partial^i\gamma}{\partial^i y} (x,x)=\lim_{T\to 0} \frac{\partial\hat\gamma^i}{\partial^i y} (x,x),
\end{equation}
for any $i=0,...,n$, and where $\gamma$ and $\hat\gamma$ are defined as in \ref{gamma2} and \ref{hatgamma}, respectively.
\end{definition}

\begin{remark}
Notice that, as $n$ increases, $\hat\gamma$ is expected to be closer to $\gamma$ (and $\hat\rho$ closer to $\rho$) for short-term and near-the-money options. 
\end{remark}

\section{The construction of optimal strike conventions}

We will make use of the following hypotheses.

\vspace{0.2cm}
\begin{description}
\item[(H1)] For any $x\in\mathbb{R}$, $k_X(x,x)=k_Y(x,x)=x$.

\item [(H2)] $\sigma\in\mathbb{L}^{2,4}$.

\item [(H3)] There exist two positive constants $a$ and $b$ such that, for any $t\in[0,T]$, $a<\sigma_t<b$.

\item [(H4)] Hypothesis (H2) holds and there exists a positive  constant  $C>0$ such that, for any $%
0<r<s<T$,
\[
E_{r}\left[ D_{r}(\sigma^i _{s})^{2}\right]\leq C, \quad i=X,Y.
\]
\item [(H5)] Hypotheses (H2) and (H4) hold and, for any $t\in [0,T]$, there exists a constant $D^+\sigma^i_0$ such that as $T\rightarrow 0$,
$$
\sup_{r,s\in [0,T]}E_t |D_s\sigma^i_r-D^+\sigma^i_0|\to 0, \quad i=X,Y.
$$
\end{description}

We note that we choose (H3) and (H4) for the sake of simplicity, but these hypotheses can be substituted by adequate integrability conditions. On the other hand, (H2) and (H5) are satisfied by the classical stochastic volatility models, where the volatility is assumed to be a diffussion process (see for example Al\`{o}s and Ewald (2008) for the Heston case). In the case of fractional volatility models with $H<\frac12$ (see for example Al\`{o}s, Le\'{o}n and Vives (2007), Fukasawa (2011) or Bayer, Friz and Gatheral (2016)), (H5) is not satisfied.  Adapting our results to these models is left for future research.

\vspace{2mm}

Our first result establishes that all the strike conventions satisfying hypotheses (H1)-(H5) are $0$-STOSCs.
\begin{proposition}
\label{propo1}
\label{gamma} Consider the model (\ref{themodel}) and assume that $(k_X,k_Y)$ is a strike convention such that hypotheses (H1)-(H5) hold. Then $(k_X,k_Y)$ is a $0$-STOSC.
\end{proposition}
\begin{proof}It suffices to see that $\lim_{T\to 0} \gamma(x,x)=\lim_{T\to 0} \hat\gamma(x,x)$. This proof will be decomposed into two steps.

{\it Step 1} Let us prove that
\begin{equation}
\label{lg}
\lim_{T\to 0} \gamma(x,x)=\tilde\sigma_0
\end{equation}
It is well known (see for example Durrleman (2008)) that the vanilla at-the-money implied volatilities $I_X, I_Y$ tend to the corresponding spot volatility. That is, 
$$
\lim_{T\to 0} \left(I_i(x,x)-E\left(\sigma^i_0\right)\right)=0,\hspace{0.2cm} i=X,Y.
$$
Now, taking into account (H1) and the fact that $\sigma^X$ and $\sigma^Y$ are right-continuous processes it follows that
$$
\lim_{T\to 0} I_i(x,k_i)=\sigma^i_0,\hspace{0.2cm}  i=X,Y,
$$
where $k_i=k_i(x,x)$. Now, as 
$$
\gamma(x,y):=\sqrt{ I_X^2(x,k_X)+ I_Y^2(y,k_Y)-2\rho I_X(x,k_X)I_Y(y,k_Y)},
$$
(\ref{lg}) follows.

{\it Step 2} Let us see that
\begin{equation}
\label{lhg}
\lim_{T\to 0} \hat\gamma(x,x)=\tilde\sigma_0.
\end{equation}
By its definition, we have that
$$
\hat\gamma(x,x)=BS^{-1}(0,x,y,V_0),
$$
where $BS^{-1}$ is the inverse of the Black-Scholes function in the sense that $$V_0=BS(0,x,y,BS^{-1}(0,x,y,V_0)).$$
Then, Theorem 5 in Al\`{o}s and Le\'{o}n (2016) gives us that 
$$V_0=E\left(BS(0,x,x,\tilde{v}_{0})\right)+o(1),$$
 which implies that
\begin{equation}
\label{hat}
\hat\gamma(x,x)=BS^{-1}\left(E\left(BS(0,x,x,\tilde{v}_{0})\right)  +o(1)\right).
\end{equation}
Moreover, the martingale representation theorem gives us that
$$
E\left(BS(0,x,x,\tilde{v}_{0})\right)=BS(0,x,x,\tilde{v}_{0})+\int_0^T A(T,s) dZ_s,
$$
for some adapted and square integrable process $A(T,\cdot)$.
This, jointly with (\ref{hat}) gives us that
\begin{eqnarray}
\lim_{T\to 0}\hat\gamma (x,x)
&=&\lim_{T\to 0} BS^{-1}\left(0,x,x,\left(BS(0,x,x,\tilde{v}_{0})+\int_0^T A(T,s) dZ_s+o(1)\right)\right)\nonumber\\
&=&\lim_{T\to 0} BS^{-1}\left(0,x,x,\left(BS(0,x,x,\tilde{v}_{0}\right)\right)\nonumber\\
&=&\lim_{T\to 0} \tilde{v}_0\nonumber\\
&=&\tilde\sigma_0,
\end{eqnarray}
and this allows us to complete the proof
\end{proof}

In order to identify the strike conventions that are $1$-STOCs, we will need the following result (see  Al\`{o}s, Le\'{o}n and Vives (2007)).

\begin{theorem}
\label{J1}
Consider the model (\ref{themodel}) and assume that hypotheses (H1)-(H5) hold. Then, for $i=X,Y$,
$$
\lim_{T\to 0}\frac{\partial I_i}{\partial z}=\frac{\rho_i D^+\sigma^i_t}{2\sigma^i_t}=\frac1{4\sigma^3_0 }\lim_{T\to 0}\frac{\langle \log S^i,M^i\rangle_T-\langle \log S^i,M^i\rangle}{T}
$$
\end{theorem}
\begin{proof}
Theorem 6.3 in Al\`{o}s, Le\'{o}n and Vives (2007) gives us that 
$$\frac{\partial I_i}{\partial y}=-\frac{\rho_i D^+\sigma^i_0}{2\sigma^i_0}.$$
for $i=X,Y$.  Now, as $\frac{\partial I_i}{\partial z}=-\frac{\partial I_i}{\partial z}$, the first equality follows. For the second one, notice that Clark-Ocone formula (see for example Nualart (2005)) gives us that
$$
(\sigma^i_t)^2=E\left((\sigma^i_t)^2\right)+\int_0^t E_r\left(D_r (\sigma^i_t)^2\right)dZ_r, \quad i=X,Y,
$$
from where we can easily deduce that
$$
dM_t^i=\left(\int_t^T E_t \left(D_t (\sigma^i_u)^2\right)du\right)dZ_t=2\left(\int_t^T E_t (\sigma^i_u D_t \sigma^i_u)du\right)dZ_t, \quad i=X,Y,
$$
from where the second equality holds.
\end{proof}
\begin{remark}
The above result gives us that the derivatives $\frac{\partial I_i}{\partial z}, i=X,Y$ depend only on the quadratic covariation between $M$ and $\log_S^i$ and on the volatility $\sigma^i$.
\end{remark}
Define $\frac{d\hat P}{d P}=e^{Y_T-Y_0}$.  The set $\mathbb{D}_{\hat Z}^{1,2}$ will denote the domain of the derivative operator $\hat D$ under $\hat P$, with respect to $\hat Z$.  We will write $\mathbb{L}_{\hat Z}^{1,2}=L^2([0,T],\mathbb{D}_{\hat Z}^{1,2})$.  Notice that as $T\to 0$, $\sup_{r,s\in [0,T]}\hat E_t |\hat D_s\sigma^i-D^+\sigma^i_0|\to 0$, for $i=X,Y$, where $D^+\sigma_0^i$ are defined as in (H5) and $\hat E_t$ is the conditional expectation with respect to $\hat P$.  

\begin{theorem} 
\label{conJorge}
Consider the model (\ref{themodel}) and assume that $\tilde\sigma\in\mathbb{L}_{\hat Z}^{1,2}$.
Then
\begin{equation}
\label{8}
\lim_{T\to 0}\frac{\partial \hat\gamma }{\partial y}(x,x)= \frac{\rho_X\sigma_0^X-\rho_Y \sigma_0^Y}{2\tilde\sigma^3}\left[ D^+\sigma_0^X(\sigma_0^X-\rho \sigma_0^Y)+ D^+\sigma_0^Y(\sigma_0^Y-\rho \sigma_0^X)\right].
\end{equation}
\end{theorem}

\begin{proof}
We have that
\begin{equation}
\label{vbs1}
V_0=BS(0,x,y,\hat\gamma).
\end{equation}
On the one hand, a direct computation gives us that
\begin{equation}
\label{vbs2}
BS(0,x,y,\hat\gamma)=e^y BS(0,x-y,0,\hat\gamma)
\end{equation}
On the other hand
\begin{eqnarray}
\label{vbs}
&&V_0=E\left(e^{X_T}-e^{Y_T}\right)^+\nonumber\\
&&=e^{Y_0}\hat{E}\left(e^{X_T-Y_T}-1\right)^+\nonumber\\
&&=e^{y}\hat{E}\left(e^{X_T-Y_T}-1\right)^+
\end{eqnarray}
where $\hat{E}$ denotes the expectation with respect to the probability measure $\hat{P}$. Notice that, under $\hat{P}$, the process $U_t:=e^{X_t-Y_t}$ satisfies  
$$dU_t=U_t  (\sigma^X_t d\hat{W}_t^X-\sigma^Y_t d\hat{W}_t^Y),$$ 
where $\hat{W}^X,\hat{W}^Y$ are $\hat{P}$-Brownian motions.
Then, (\ref{vbs2}) and (\ref{vbs}) gives us that (\ref{vbs1}) is equivalent to
$$
\hat{E}(U_T-1)^+=BS(0,x-y,0,\hat\gamma).
$$
Notice that $\hat\gamma$ is the implied volatility of a vanilla option with strike 1 on an underlying $U_t$, with volatility $\tilde\sigma$.  
Then Theorem \ref{J1} gives us that
$$\lim_{T\to 0} \frac{\partial \hat\gamma}{\partial z}(x,x) = \frac1{4\tilde\sigma_0^3 T}\lim_{T\to 0} \frac{\langle U,\tilde M\rangle_T}{T}.$$
Now, as  
$$d\tilde M_t=\left(\int_t^T \hat E_r (\hat D_r\tilde\sigma_t^2)dr\right)d\hat Z_t$$
and 
$$\hat D_r\tilde\sigma_t^2=2\sigma_t^X \hat D_r\sigma_t^X + 2\sigma_t^Y \hat D_r\sigma_t^Y-2\rho\sigma_t^X \hat D_r\sigma_t^Y-2\rho\sigma_t^Y \hat D_r\sigma_t^X.$$
we get that 
$$\frac1{4\tilde\sigma_0^3 T} \lim_{T\to 0} \langle U,\tilde M \rangle=  \frac{\rho_X\sigma_0^X-\rho_Y\sigma_0^Y}{2\tilde\sigma_0^3}\left[ D^+\sigma_0^X(\sigma_0^X-\rho \sigma_0^Y)+D^+\sigma_0^Y(\sigma_0^Y-\rho \sigma_0^X)\right].$$
This completes the proof.
\end{proof}

%\vspace{10cm}
%%Notice that $\lambda_1 \hat{W}_t^1-\lambda_2 \hat{W}_t^2$ is a Gaussian noise under $\hat{P}$. Then the volatility of $Y$ is given by $\bar{v}$ and  
%\begin{eqnarray}
%\langle\hat{M},\ln Y\rangle_t&&=C^2(\lambda_1,\lambda_2,\rho)D^+\sigma_t^2\left(\lambda_1\rho_1-
%\lambda _{2}\rho_2\right)\nonumber \\
%&&=2\sigma_t C^2(\lambda_1,\lambda_2,\rho)D^+\sigma_t\left(\lambda_1\rho_1-
%\lambda _{2}\rho_2\right).
%\end{eqnarray}
%Then
%\begin{eqnarray}
%\label{implicit2}
%\frac{\partial\hat\gamma}{\partial y}&=&\frac{2\sigma_t}{4v^3_t} C^2(\lambda_1,\lambda_2,\rho)D^+\sigma_t\left(\lambda_1\rho_1-
%\lambda _{2}\rho_2\right)\nonumber\\
%&=&\frac{ D^+\sigma_t}{2C\left( \lambda _{1},\lambda _{2},\rho \right)\sigma_t }\left(\lambda_1\rho_1-
%\lambda _{2}\rho_2\right),
%\end{eqnarray}
%as we wanted to proof.
%%\end{proof}

\vspace{2mm}

In the next theorem we establish a condition for a strike convention to be a $1$-STOSC. This is the main result of this paper.

\begin{theorem}
\label{firstorder}
Consider the model (\ref{themodel}) and assume that hypotheses (H1)-(H5) hold. Then, a strike covention $(k_1,k_2)$ is a $1$-STOSC if and only if
\begin{eqnarray}
\label{conve-general}
&&\frac{\rho_X\sigma_0^X-\rho_Y \sigma_0^Y}{2\tilde\sigma_0^3}\left[ D^+\sigma_0^X(\sigma_0^X-\rho \sigma_0^Y)+ D^+\sigma_0^Y(\sigma_0^Y-\rho\sigma_0^X)\right]\\
&&\quad =\lim_{T\to 0}\left\{\frac{1}{\gamma }\left[ I_X\frac{\partial I_X}{\partial z}\frac{\partial k_X}{\partial y}
+I_Y\left( \frac{\partial I_Y}{\partial z}\frac{\partial k_Y}{\partial y}+\frac{\partial I_Y}{\partial y}\right)\right.\right. \nonumber \\
&&\qquad\qquad\qquad\left.\left.-\rho I_X\left( \frac{\partial I_Y}{\partial z}\frac{\partial k_Y}{\partial y}+\frac{\partial I_Y}{\partial y}\right) 
-\rho I_Y\frac{\partial I_X}{\partial z}\frac{\partial k_X}{\partial y}\right](x,x) \right\}\nonumber
\end{eqnarray}
\end{theorem}

 \begin{proof}
We have to prove that
$$
\lim_{T\to 0} \left(\frac{\partial \gamma }{\partial y}-\frac{\partial \hat\gamma }{\partial y}\right)(x,x)=0.
$$
Theorems \ref{J1} and \ref{conJorge} directly give us the desired result.
%A direct computation gives us that
%\begin{eqnarray}
%\label{fromderivatives}
%&&\frac{\partial \gamma }{\partial y}\nonumber\\
%&&=\frac{1}{2\gamma }\left[ 2I_{1}%
%\frac{\partial I_{1}}{\partial z}\frac{\partial k_{1}}{\partial y}%
%+2I_{2}\left( \frac{\partial I_{2}}{\partial z}\frac{\partial k_{2}}{%
%\partial y}+\frac{\partial I_{2}}{\partial y}\right) -\rho I_{1}\left( \frac{%
%\partial I_{2}}{\partial z}\frac{\partial k_{2}}{\partial y}+\frac{\partial
%I_{2}}{\partial x}\right) -\rho I_{2}\frac{\partial I_{1}}{\partial z}\frac{%
%\partial k_{1}}{\partial y}\right]  \nonumber\\
%&&=\frac{1}{C\left( \lambda _{1},\lambda _{2},\rho \right) }\left[
%\left( \lambda _{1}-\rho \lambda _{2}\right) \frac{\partial I_{1}}{\partial z%
%}\frac{\partial k_{1}}{\partial y}+\left( \lambda _{2}-\rho \lambda
%_{1}\right) \left( \frac{\partial I_{2}}{\partial z}\frac{\partial k_{2}}{%
%\partial y}+\frac{\partial I_{2}}{\partial y}\right) \right].
%\end{eqnarray}%
%On the other hand, Theorems \ref{J1} and \ref{conJorge}  gives us that
%
%\begin{equation}
%\label{88}
%\lim_{T\to 0}\frac{\partial \hat\gamma }{\partial y}(x,x)=\frac{1}{C\left( \lambda _{1},\lambda _{2},\rho \right)}\left(\lambda_1\frac{\partial I_{1}}{\partial z}-\lambda _{2}\frac{\partial I_{2}}{\partial z}\right)(x,x).
%\end{equation}
%Then,  (\ref{fromderivatives}) and (\ref{88}) imply the desired result.
\end{proof}

\begin{remark}
If $\rho_X\ne 0$ and $\rho_Y\ne 0$, then $D^+\sigma_0^i =\frac{\sigma_0^i}{\rho_i}\lim_{T\to 0}\frac{\partial I_i}{\partial z}$ for $i=X,Y$.  Then the left hand side in \eqref{conve} can be written as
\begin{equation}
\label{rhoXrhoYremark}
\lim_{T\to 0}\frac{\rho_X I_X -\rho_Y I_Y}{\gamma^3}\left[\frac{\partial I_X}{\partial z}\frac{I_X}{\rho_X}(I_X-\rho I_Y)+\frac{\partial I_Y}{\partial z}\frac{I_Y}{\rho_Y}(I_Y-\rho I_X)\right](x,x).
\end{equation}
This gives us a model-free condition for a $1$-STOSC, in the sense that a specific model for the volatility processes is not needed. 
\end{remark}

While various different cases of the general rule above may be considered, a convenient particular case of a strike convention is obtained if $\sigma^X_t=\lambda_X \sigma_t$ and $\sigma^Y_t=\lambda_Y \sigma_t$, where $\lambda_X$ and $\lambda_Y$ are positive constants and $\sigma_t$ is a non-negative, right-continuous and square integrable process adapted to the filtration generated by $Z_t$.   We note that this case of a single volatility process shifted by a constant for each of the two assets is a generalization of the model introduced for correlation options in Bakshi and Madan (2000) and also for spread options in Dempster and Hong (2000).  For convenience we shall refer to this model as the one-volatility two-levels (1V2L) model.  The following corollary demonstrates that for the 1V2L model a strike convention can be derived either in terms of model parameters or market observables, namely the short-time limits of the corresponding vanilla implied volatility levels and skews.   
\begin{corollary}
\label{STOSCcorollary}
Assume the 1V2L model.  Then, a strike covention $(k_1,k_2)$ is a $1$-STOSC if and only if
\begin{eqnarray}
\label{convention}
&&\lim_{T\to 0}\left[
\left( 1-\rho \frac{I _Y}{I_X}\right) \frac{\partial I_{X}}{\partial z%
}\frac{\partial k_{X}}{\partial y}+\left( \frac{I _{Y}}{I_X}-\rho \right) \left( \frac{\partial I_{Y}}{\partial z}\frac{\partial k_{Y}}{%
\partial y}+\frac{\partial I_{Y}}{\partial y}\right)\right. \\
&&\qquad-\left.\frac{\partial I_{X}}{\partial z}+\frac{I _{Y}}{I_X}\frac{\partial I_{Y}}{\partial z}\right](x,x)= 0. \nonumber
\end {eqnarray}
or equivalently (in terms of model parameters):
\begin{equation}
\label{conve}
\left(1-\rho\frac{\lambda_Y}{\lambda_X}\right)\frac{\rho_X}{\rho_Y} \frac{\partial k_X}{\partial y} + \left(\frac{\lambda_Y}{\lambda_X}-\rho\right)\left(\frac{\partial k_Y}{\partial y}-1\right) = \frac{\rho_X}{\rho_Y}-\frac{\lambda_Y}{\lambda_X}
\end{equation}
\end{corollary}
\begin{proof}
In the 1V2L model (with $\sigma_t^X=\lambda_X\sigma_t$, $\sigma_t^Y=\lambda_Y\sigma_t$), Theorem \ref{J1} implies that
$$\frac{\partial I_Y}{\partial z}=\frac{\rho_Y}{\rho_X}\frac{\partial I_X}{\partial z}.$$
Expanding the expression in \eqref{rhoXrhoYremark} and substituting for $\frac{\partial I_Y}{\partial z}$ we obtain
\begin{eqnarray*}
%&&\lim_{T\to 0}\left\{\frac{1}{\gamma }\left[ I_X\frac{\partial I_X}{\partial z}\frac{\partial k_X}{\partial y}
%+I_Y\left( \frac{\partial I_Y}{\partial z}\frac{\partial k_Y}{\partial y}+\frac{\partial I_Y}{\partial y}\right)\right.\right. \nonumber \\
%&&\qquad\qquad\qquad\left.\left.-\rho I_X\left( \frac{\partial I_X}{\partial z}\frac{\partial k_Y}{\partial y}+\frac{\partial I_Y}{\partial x}\right) 
%-\rho I_Y\frac{\partial I_X}{\partial z}\frac{\partial k_X}{\partial y}\right](x,x) \right\}\nonumber
&=&\lim_{T\to 0}\left[\frac{1}{\gamma^3}\left(\rho_XI_X-\rho_YI_Y\right)\frac{\partial I_X}{\partial z} \frac1{\rho_X} \left(I_X (I_X-\rho I_Y)+I_Y(I_Y-\rho I_X)\right)\right](x,x)\\
&=&\lim_{T\to 0}\left[\frac{1}{\gamma^3}\left(I_X\frac{\partial I_X}{\partial z}-I_Y\frac{\rho_Y}{\rho_X}\frac{\partial I_X}{\partial z}\right) \left(I_X (I_X-\rho I_Y)+I_Y(I_Y-\rho I_X)\right)\right](x,x)\\
&=&\lim_{T\to 0}\left[\frac{1}{\gamma^3}\left(I_X\frac{\partial I_X}{\partial z}-I_Y\frac{\partial I_Y}{\partial z}\right) \gamma^2\right](x,x)\\
&=&\lim_{T\to 0}\left[\frac{1}{\gamma}\left(I_X\frac{\partial I_X}{\partial z}-I_Y\frac{\partial I_Y}{\partial z}\right)\right](x,x)
\end{eqnarray*}
Equating this with the right hand side of \eqref{conve-general} and rearranging gives the desired result.  The equivalent result in terms of model parameters $\rho_X,\rho_Y,\lambda_X,\lambda_Y$ can be found by Theorem \ref{J1} and the fact that the at-the-money (ATM) implied volatility tends to the corresponding spot volatilty at time zero.
\end{proof}

\section{Optimal linear log-strike conventions}

Several strike conventions have been proposed in the literature. Some
classical examples (see for example Alexander and Venkatramanan (2011)
and Swindle (2014)) are of the form%
\begin{equation}
\left\{ 
\begin{array}{c}
k_{X}(x,y)=(1-a)x+ay \\ 
k_{Y}(x,y)=ax+(1-a)y%
\end{array}%
\right. ,  \label{linearstrikeconv}
\end{equation}%
for some real number $a$. For example, in Swindle (2014)\ the authors
suggest to take $k_{X}=\ln S_{t}^{Y}$ and $k_{Y}=\ln S_{t}^{X}.$ This choice
corresponds to (\ref{linearstrikeconv}) in the case $a=1.$ 
On the other hand, in Alexander and Venkatramanan (2011) the authors
mostly study the strike convention $k_{X}=\ln S_{t}^{X}$ and $k_{Y}=\ln
S_{t}^{Y}$, which is the case $a=0$. 
In this section we will find an optimal linear log-strike option of the form
(\ref{linearstrikeconv}). Given two strikes $k_{X},k_{Y}$ of the form (\ref%
{linearstrikeconv}), we have 
$$\frac{\partial k_X}{\partial y} = a, \quad \frac{\partial k_Y}{\partial y}=1-a,$$
and thus equation (\ref{conve}) reduces to 
%\begin{equation}
%\label{lasuperformulalinear}
%a\frac{\partial I_{1}}{\partial z}\left(
%1-\frac{\rho\lambda _{2}}{\lambda_1}\right) +(1-a)%
%\frac{\partial I_{2}}{\partial z}\left( \frac{\lambda _{2}}{\lambda_1}-\rho \right)  =\frac{\partial I_{1}}{\partial z} -\rho \frac{\partial I_{2}}{\partial z}.
%\end{equation}%
%That is,
%\begin{equation}
%\label{lasuperformulalinear2}
%a\left[\frac{\partial I_{1}}{\partial z}\left(1-\frac{\rho\lambda _{2}}{\lambda_1}\right) -\frac{\partial I_{2}}{\partial z}\left( \frac{\lambda _{2}}{\lambda_1}-\rho \right)\right]
% =\frac{\partial I_{1}}{\partial z}-\frac{\lambda _{2}}{\lambda_1}\frac{\partial I_{2}}{\partial z}.
%\end{equation}%
%Assume that $\rho_2\neq 0$. Then  Theorem \ref{J1} imply that $\frac{\partial I_{1}}{\partial z}=\frac{\rho_1}{\rho_2}\frac{\partial I_{2}}{\partial z}$ and then (\ref{lasuperformulalinear2}) becomes

\begin{equation}
\label{lasuperformulalinear3}
a\left[\frac{\rho_X}{\rho_Y}\left(1-\frac{\rho\lambda _{Y}}{\lambda_X}\right) -\left( \frac{\lambda _{Y}}{\lambda_X}-\rho \right)\right]
 =\frac{\rho_X}{\rho_Y}-\frac{\lambda _{Y}}{\lambda_X}.
\end{equation}%
or, alternatively,
$$
a\left[\rho_X(\lambda_X-\rho\lambda_Y)-\rho_Y(\lambda_Y-\rho\lambda_X)\right]
 =\lambda_X\rho_X-\rho_Y\lambda_Y.
$$
Then, if 
$$
[\rho_X(\lambda_X-\rho\lambda_Y)-\rho_Y(\lambda_Y-\rho\lambda_X)\neq 0
$$
there exists a unique $1$-STOSC, given by $a=a^\star$, where
\begin{equation}
\label{theoptimal}
a^\star=\frac{\rho_X\lambda_X-\rho_Y\lambda_Y}{\rho_X(\lambda_X-\rho\lambda_Y)-\rho_Y(\lambda_Y-\rho\lambda_X)}
\end{equation}

\begin{remark}
We note several interesting special cases related to this result:
\begin{enumerate}
\item The underlying prices $S^X$ and $S^Y$ are uncorrelated ($\rho=0$):
$$a^\star=\frac{\rho_X\lambda_X-\rho_Y\lambda_Y}{\rho_X \lambda_X-\rho_Y\lambda_Y}=1$$
Intuitively, thinking of an exchange option as a regular option with floating strike, if the strike is uncorrelated, then it is optimal to use the implied volatility corresponding to that floating strike (to the opposite leg of the spread), the `volatility look-up heuristic' of Swindle (2014).  
\item The two volatilities have the same level ($\lambda_X=\lambda_Y$):
$$a^\star=\frac{\rho_X-\rho_Y}{\rho_X(1-\rho)-\rho_Y(1-\rho)}=\frac{1}{1-\rho}$$
Notice that in this case $a^\star$ is no longer dependent on $\rho_X,\rho_Y,\lambda_X,\lambda_Y$.
\item The two asset to volatility correlations are equal ($\rho_X=\rho_Y$):
$$a^\star=\frac{\lambda_X-\lambda_Y}{\lambda_X-\rho\lambda_Y-\lambda_Y+\rho\lambda_X)}=\frac{1}{1+\rho}$$
Again, here $a^\star$ no longer depends on correlations $\rho_X,\rho_Y$ or levels $\lambda_X,\lambda_Y$.
\item Asset to volatility correlation is zero ( $\rho_Y=0$):
$$a^\star=\frac{\lambda_X\rho_X}{\rho_X(\lambda_X-\rho\lambda_Y)}=\frac{\lambda_X}{\lambda_X-\rho\lambda_Y}$$
Similarly, if $\rho_X=0$, then $a^\star=\frac{\lambda_Y}{\lambda_Y-\rho\lambda_X}$.  In these cases we can also conclude (since $\lambda_X,\lambda_Y>0$) that $\rho>0$ corresponds to $a^\star>1$ (and $\rho<0$ to $a^\star<1$), intuitive for a floating strike option in which the strike tends to move away as $S$ moves towards it.  
\end{enumerate}
\label{specialcases}
\end{remark}

\begin{remark}
Similarly to the equivalence expressions within Corollary \ref{STOSCcorollary}, we note that since $\frac{\lambda_Y}{\lambda_X}=\lim_{T\to 0}\frac{I_Y}{I_X}$ and $\frac{\rho_Y}{\rho_X}=\lim_{T\to 0}\frac{\partial I_Y/ \partial z}{\partial I_X/ \partial z}$,
%is the limit of the corresponding vanilla ATM implied volatility levels and $\frac{\rho_}{\lambda_X}$ is the limit of the corresponding vanilla ATM implied volatility skews, 
our results can be transformed from model parameters to market observables. Thus, the optimal strke convention can be computed from equation (\ref{lasuperformulalinear3}), needing only to know $\rho$ (often estimated from price histories) and the short-time limits of the corresponding vanilla implied volatility levels and skews.
\end{remark}

\newpage
\section{Numerical examples}

In order to investigate the performance of the optimal log-linear strike convention given by \eqref{theoptimal}, we consider a number of numerical examples of spread option pricing with different assumptions for parameter values.  In each case we compare our optimal choice of $a^\star$ to results from using the other common strike conventions of `at-the-money' (ATM) implied volatilities for each asset (i.e. $a=0$) or the volatility look-up heuristic (i.e. $a=1$).  As a simple and commonly-used benchmark, we use the Heston Model throughout, but conduct tests under a large variety of different parameter sets.  \\

Volatility dynamics are given by:
 $$ d\sigma_{t}^{2}   =\kappa\left(  \theta-\sigma_{t}^{2}\right)  dt+\nu\sqrt{\sigma_{t}^{2}}dZ_{t}^{(3)},$$
within the 1V2L model, a version of model \eqref{themodel} introduced before Corollary \ref{STOSCcorollary}.

\subsection{Test Cases}

For now, we consider two test cases with parameters as described below, varying only $\rho_{Y}$ between cases:
\begin{itemize}
\item option maturity:  $T=0.05$ (a few weeks)
\item volatility process ($\sigma_t$) parameters: $\kappa=1.5,\theta=0.15,\nu=0.5,\sigma_0=0.15$
\item volatility scaling factors:  $\lambda_X=1.5,\lambda_Y=1$
\item correlation parameters:  $\rho=0.5,\rho_{X}=-0.4$, and $\rho_{Y}=-0.6$ or $\rho_{Y}=0.4$
\end{itemize}
Note that Test Case 1 with $\rho_{Y}=-0.6$ corresponds to two downward-sloping implied volatility skews for the two assets, while Test Case 2 with $\rho_Y=0.4$ produces an upwards skew for the second asset.  The top row of Figure \ref{TestCases} shows these implied volatility plots, generated by pricing single asset options under the Heston model.  We then use these saved implied volatilities to price an exchange option with payoff $(S^X_T-S^Y_T)^+$ across a range of moneyness, with $S^X_0=100$ fixed and $S^Y_0\in [80,120]$.  Margrabe's formula with the three different strike conventions (choices of $a$) is compared against an `exact solution' using 1,000,000 simulated paths (with the constant volatility solution as a control variate).  \\

\begin{figure}[htbp]
\centering
\includegraphics[width=0.47\textwidth]{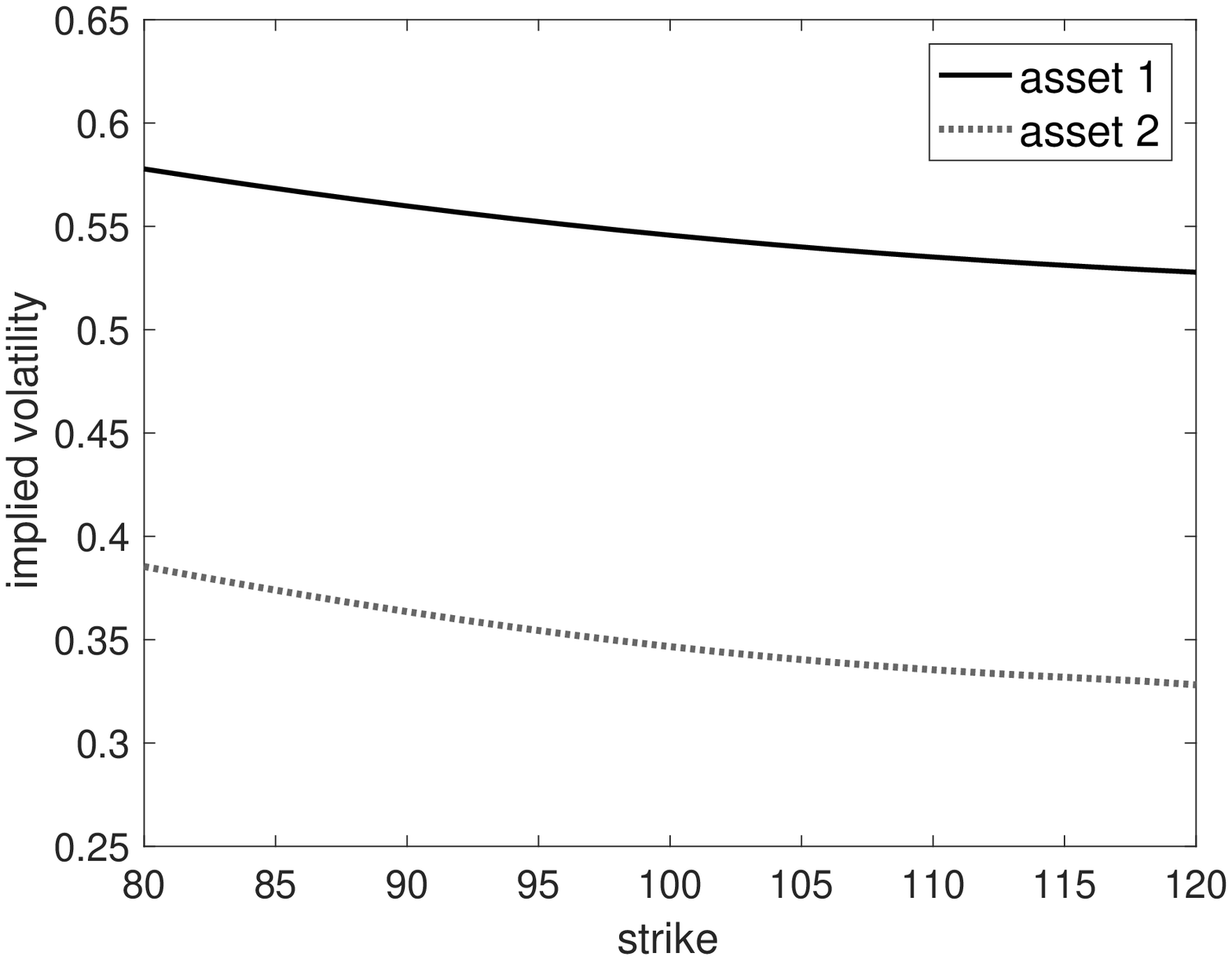}
\includegraphics[width=0.47\textwidth]{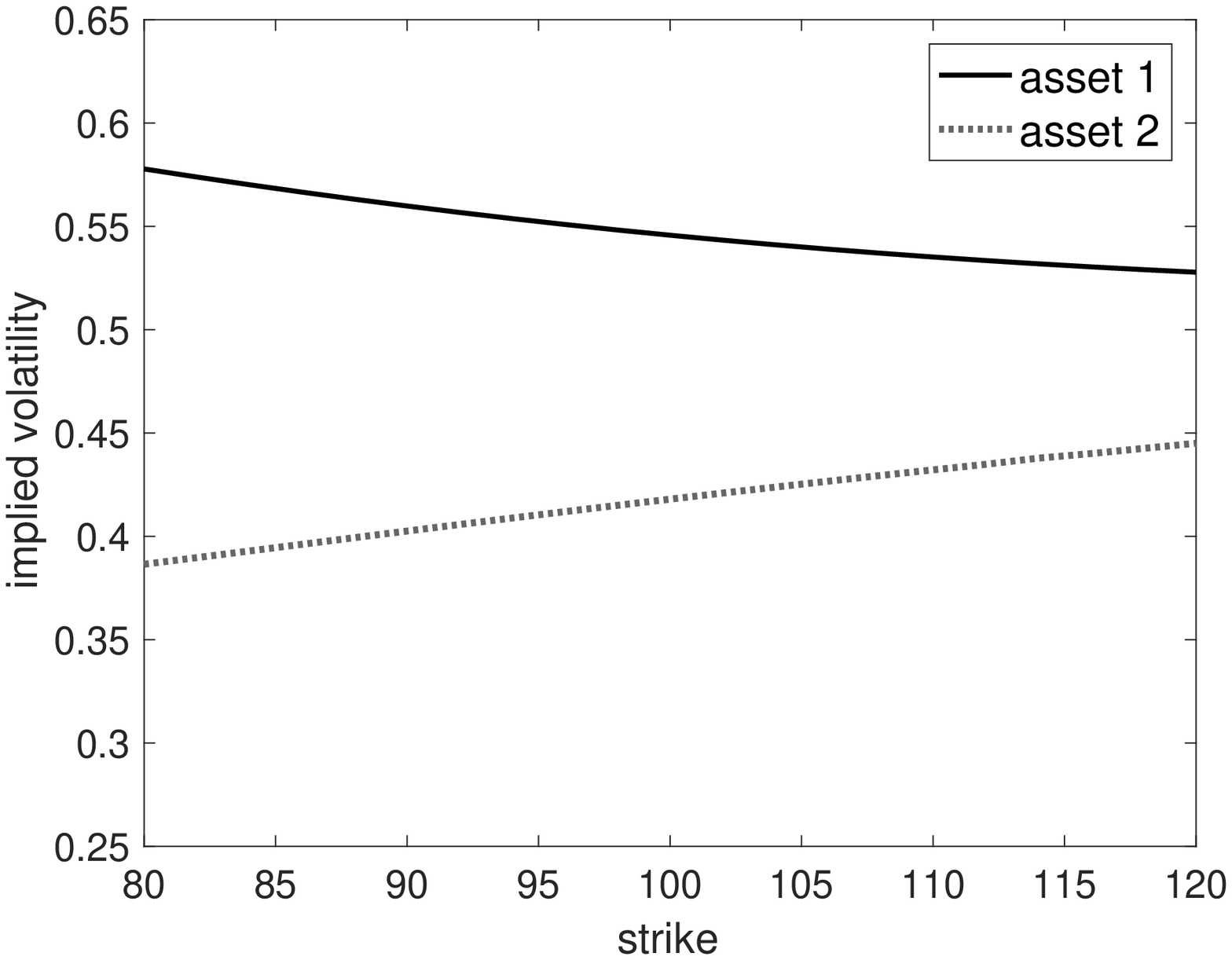}\\
\includegraphics[width=0.47\textwidth]{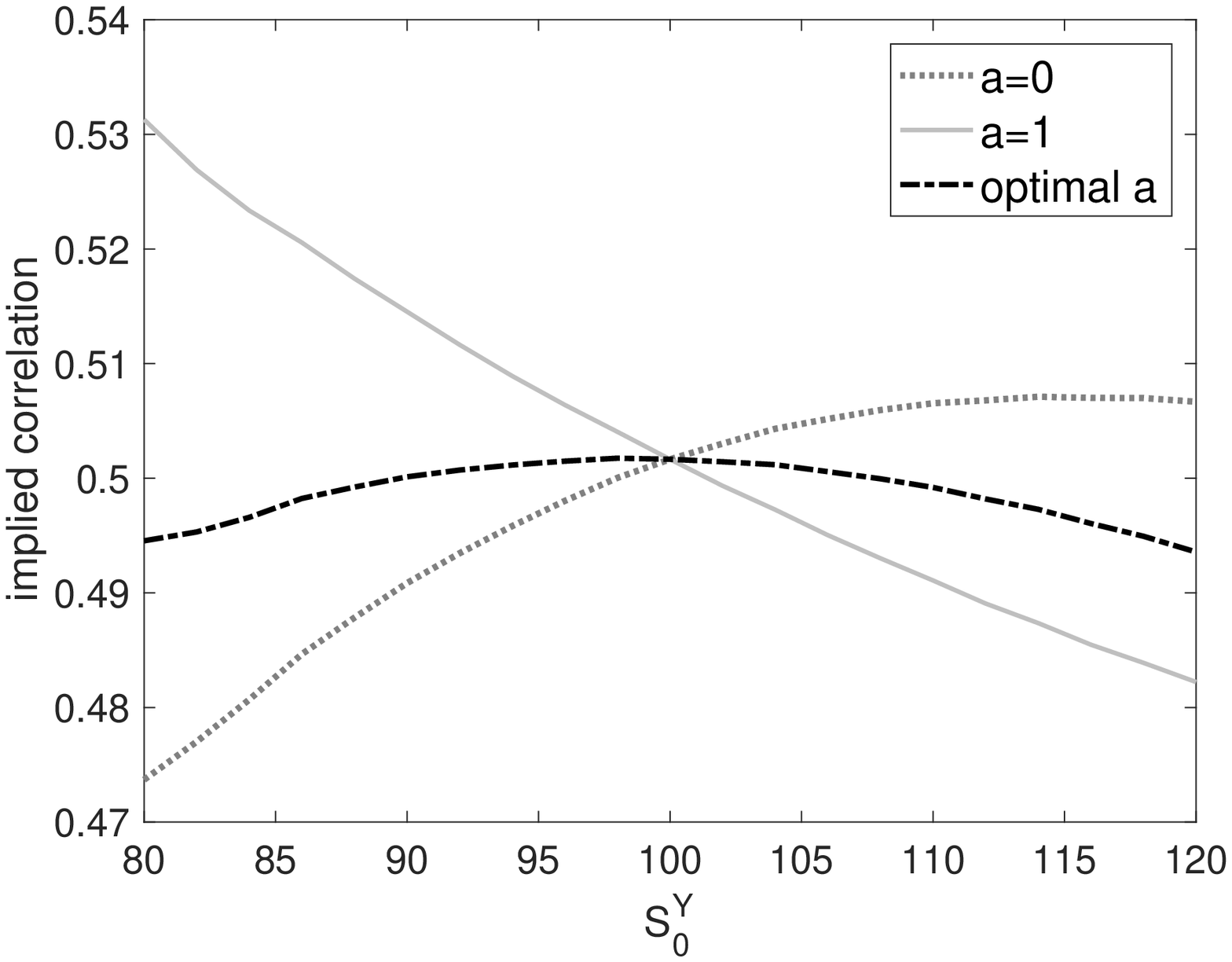}
\includegraphics[width=0.47\textwidth]{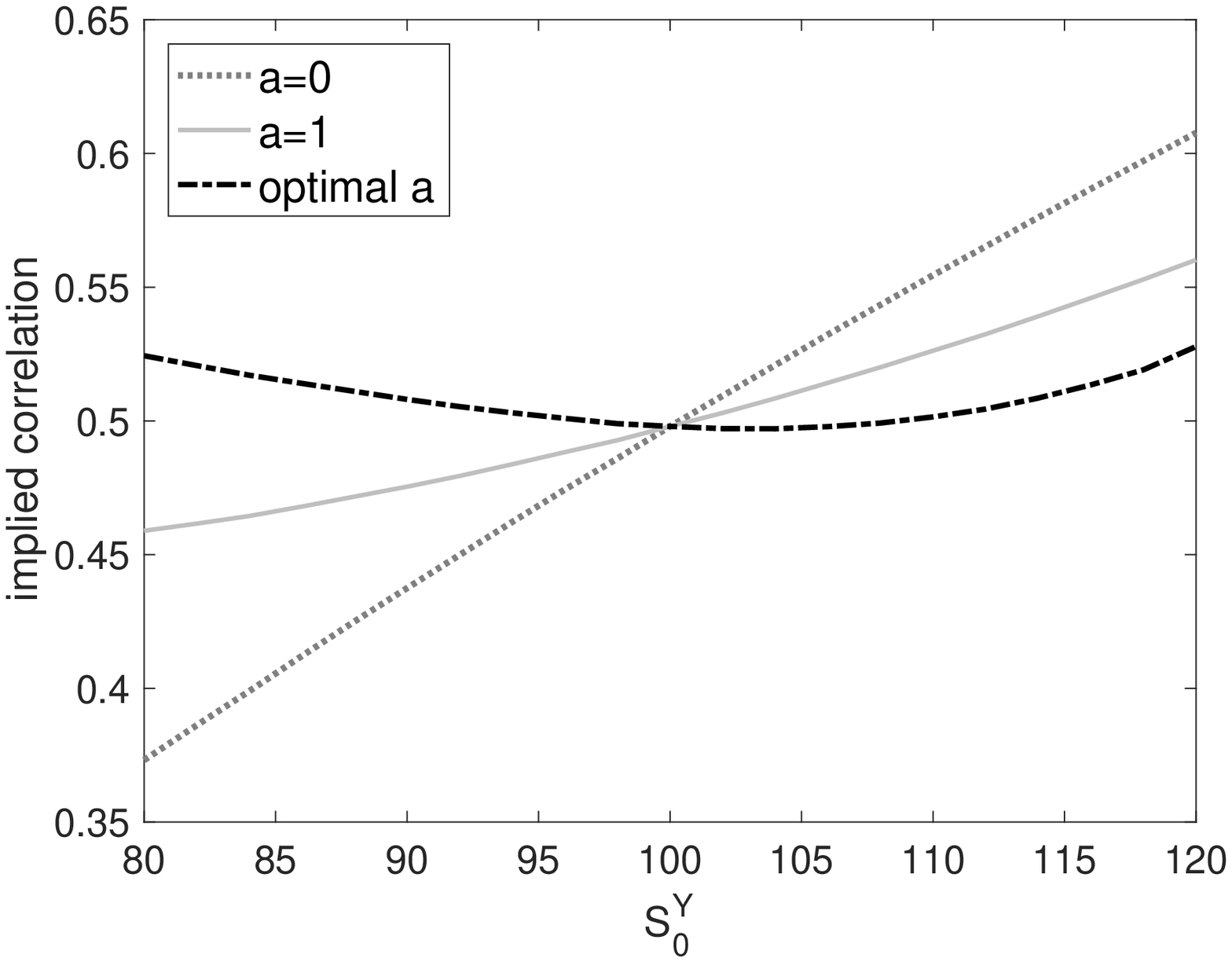}\\
\includegraphics[width=0.47\textwidth]{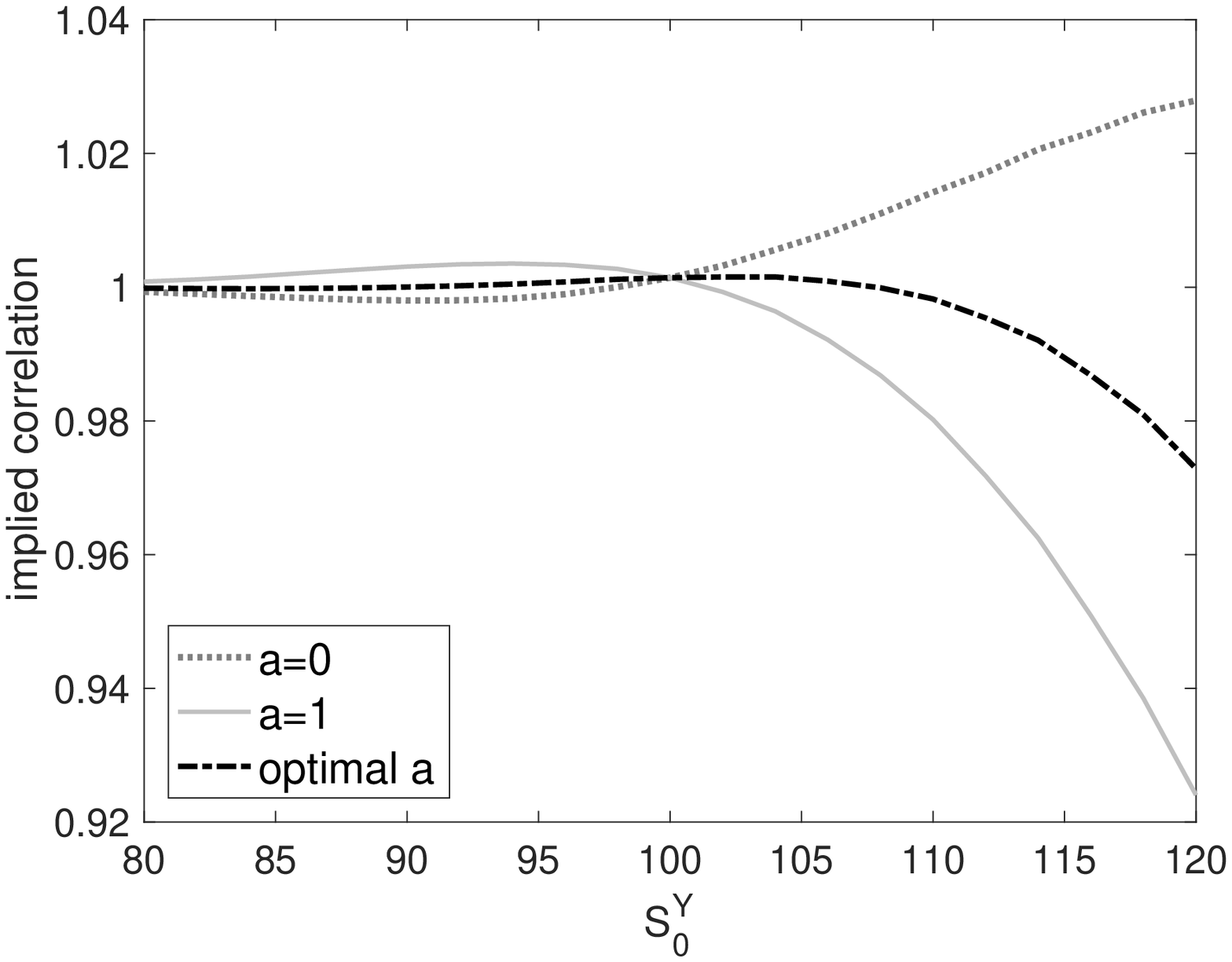}
\includegraphics[width=0.47\textwidth]{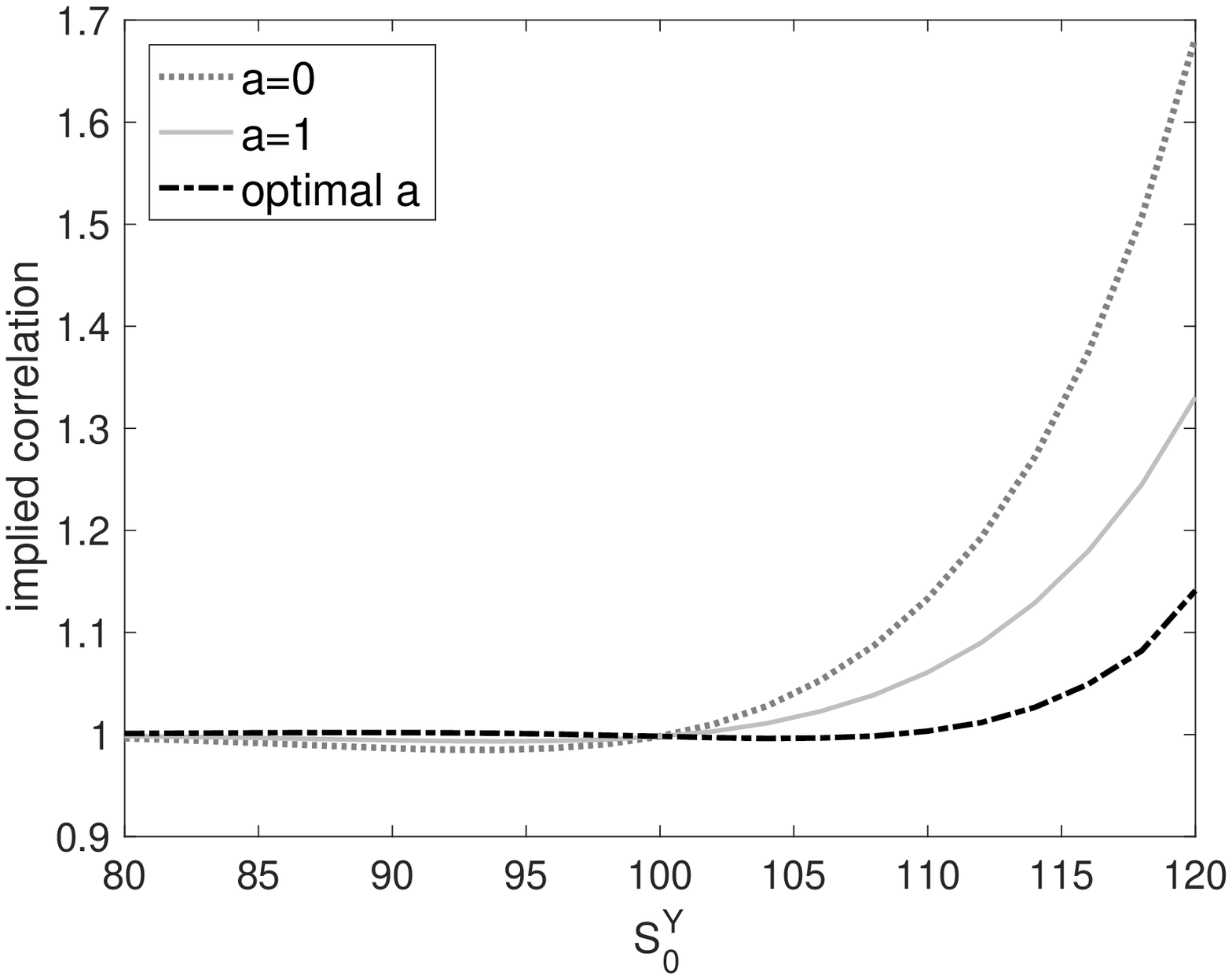}\\
\includegraphics[width=0.47\textwidth]{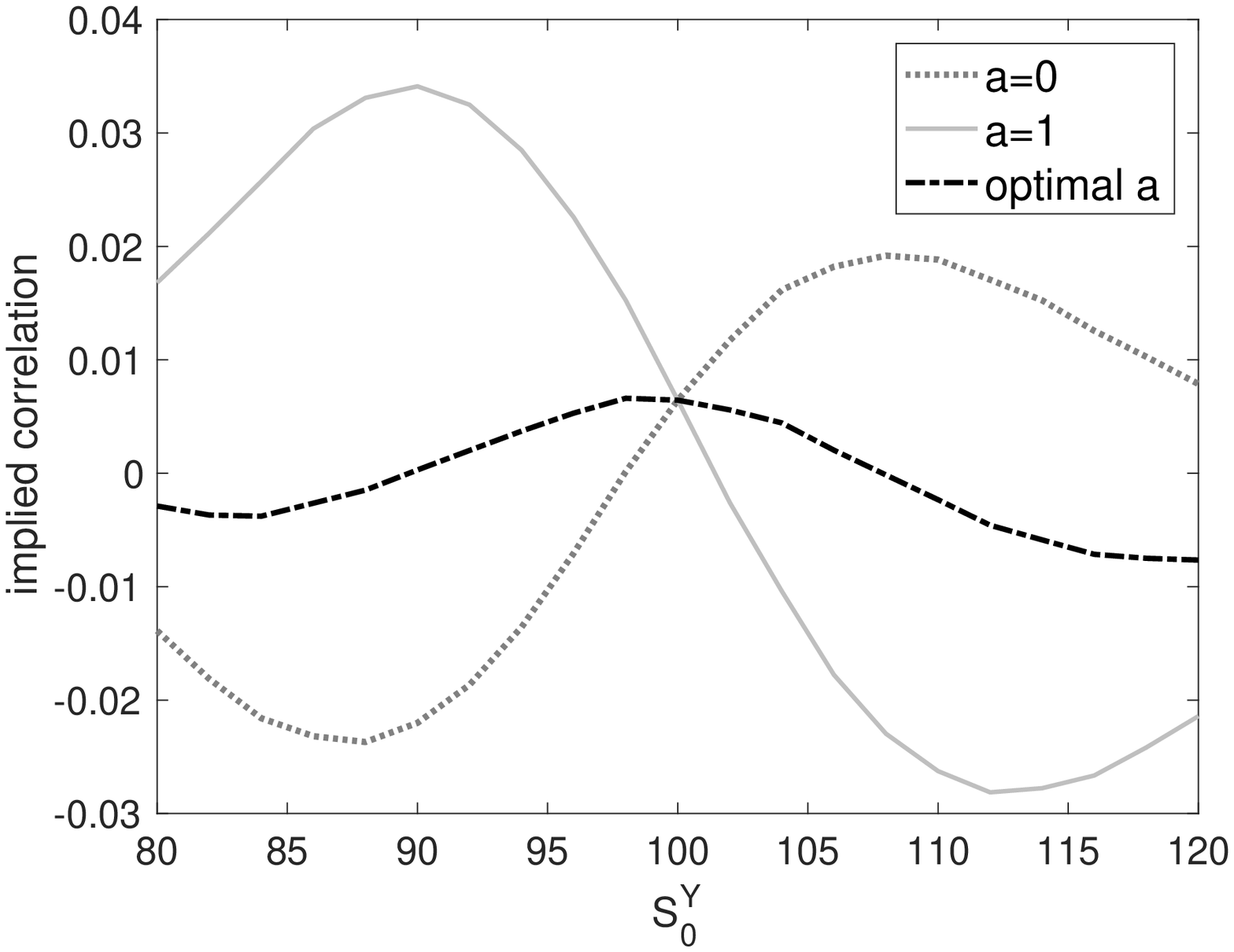}
\includegraphics[width=0.47\textwidth]{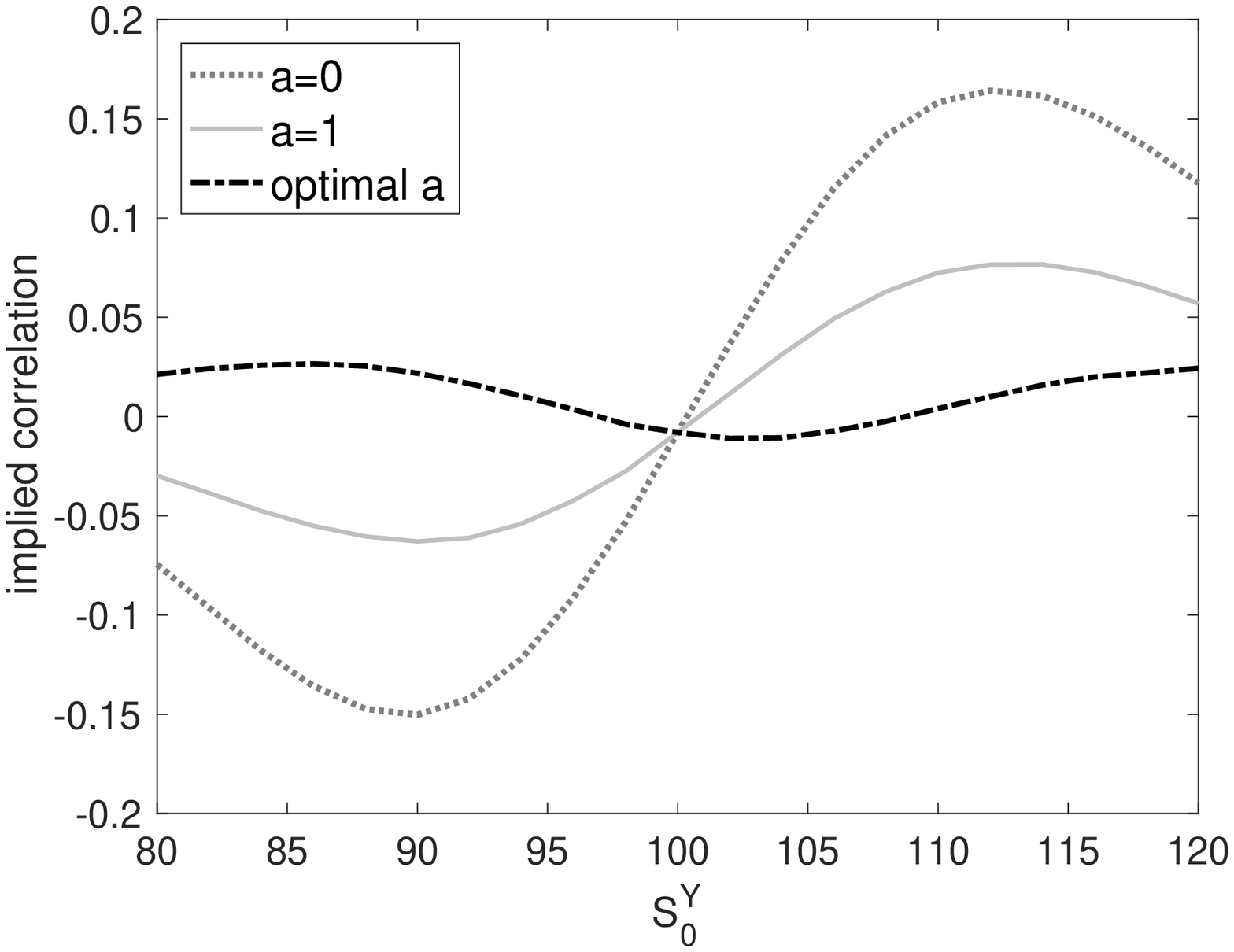}
\caption{Test case results against moneyness: implied volatility skews (first row), implied correlations (second row), spread option price ratios (third row) and price differences (fourth row).  Left column is Test Case 1 ($\rho_{Y}=-0.6$) and right column is Test Case 2 ($\rho_{Y}=0.4$).} 
\label{TestCases}
\end{figure}
 
Rows 2-4 of Figure \ref{TestCases} provide three alternative ways of visualizing the performance of our optimal strike convention (the darkest line) across moneyness compared with the other approaches: 
%\begin{enumerate}[(i)]
%\item by converting spread option prices back to implied correlations $\hat\rho$ in order to compare with the model correlation of $\rho=0.5$, recalling Remark \ref{impliedcorrelationremark};
%\item by plotting the ratio of Margrabe price to exact price; 
%\item by plotting the difference between Margrabe and exact.  
%\end{enumerate}
 (i) by converting spread option prices back to implied correlations $\hat\rho$ in order to compare with the model correlation of $\rho=0.5$, recalling Remark \ref{impliedcorrelationremark};
 (ii) by plotting the ratio of Margrabe price to exact price; 
 (iii) by plotting the difference between Margrabe and exact.  
Note that the ATM values ($S^Y_0=100$) are equal across strike conventions since they all coincide at this point, and $\rho\approx\hat{\rho}$.  However, moving away from the ATM point ($S^Y_0=100$), we can clearly see that the optimal $a^\star$ performs significantly better than the other contenders. \\

On Test Case 1 (left column) we see that $a=0$ significantly overprices the spread option ($\hat\rho<\rho$) when $S^Y_0<S^X_0$ (the `in the money', or ITM, case) and underprices ($\hat\rho>\rho$) when $S^Y_0>S^X_0$ (the `out of the money', or OTM, case), while $a=1$ does the opposite.  It might therefore appear that a rather arbitrary midpoint convention of $a=1/2$ could work as a compromise between the other rules, but this is not surprising considering that $a^\star=0.429$ is optimal in this case.  In contrast, on Test Case 2 (right column), $a^\star=1.917$ is optimal, and thus both the other strike conventions overprice ITM and underprice OTM, sometimes by a large amount.  Our approach keeps absolute errors in Case 1 below 0.01 and in Case 2 below 0.03 across different $S^Y_0$ values.  This consistent pricing of options at different moneyness levels is a major advantage.  In practice an indicative quote on a different spread option in the market could therefore more accurately be used to price another contract.  \\

Although dominated by skew, the implied correlation plots in Figure \ref{TestCases} reveal a slight `frown' in the first test case, as sometimes witnessed in the market. Ideally we would like to observe a flat line at $\hat\rho=0.5$, as the theory dictates should hold with short enough $T$ and near the money, but our results are nonetheless encouraging.  Note that when looking at relative pricing errors in the third row of plots, errors unsurprisingly dominate for OTM options which always have zero intrinsic value and much lower prices than ITM.  It is more interesting to note the patterns in the case of absolute errors just below, in particular that deep ITM and OTM options show less pricing error than moderately ITM and OTM.  This effect can be explained by the fact that there is less (model-dependent) extrinsic value to accurately price.  

\subsection{Extensive Numerical Investigations}

Instead of considering individual cases of parameter sets as above, we now test the approach across a wide range of different parameter values and in particular correlation structures.  We use the following ranges for our parameters:\footnote{Note that sometimes round numbers (and zeros) are specifically avoided due to the unrealistically large values of $|a^\star|$ they can produce in \eqref{theoptimal}. This is not unreasonable considering that data fitting rarely produces round numbers!}

\begin{itemize}
\item $T \in [0.05, 0.1, 0.25, 0.5, 1]$
\item $S_0^X=100$, $S_0^Y \in[80,84,\ldots,100,\ldots,116,120]$
\item $\lambda_X=1$, $\lambda_Y=1.24$ (note: tests for different $\lambda$s perform similarly)
\item Heston parameters (as before): $\kappa=1.5,\theta=0.15,\nu=0.5,\sigma_0=0.15$
\item $\rho\in[-0.9,-0.7,-0.5,-0.3,-0.1,0.1,0.3,0.5,0.7,0.9]$
\item $\rho_X\in[-0.72,   -0.42,   -0.12,    0.18,   0.48]$
\item $\rho_Y\in[-0.61,   -0.31,   -0.01,    0.29,    0.59]$
\end{itemize}

We shall compare results using a variety of commonly-used pricing errors such as Mean Absolute Error (MAE), Mean Absolute Percentage Error (MAPE), Root Mean Squared Error (RMSE), Maximum Absolute Error (MaxAE, i.e. worst case), as well as considering the mean standard deviation (MStd) of errors across moneyness ($S_0^Y$ grid).  The first two correspond to the price ratio and price difference plots in Figure \ref{TestCases} while the last of these is a way to assess the methodology's aim of pricing consistently across moneyness, or in other words flattening the implied correlation skew or frown we would otherwise observe. \\  

Table \ref{MAE1} shows the MAE (between simulated prices and Margrabe prices), averaging over the $S_0^Y$, $\rho_X$ and $\rho_Y$ grids\footnote{Each number in the table is thus an average of $11\times 5\times 5=275$ cases (gridpoints).}, for the different choices of $\rho$, $T$ and of course $a$.  We only show half of our $\rho$ values here as a reasonable sample.  When calculating average errors, we first exclude parameter sets which lead to a non-valid (non positive definite) correlation matrix.  This is 19.6\% of the cases overall, and around half of the cases for the most extreme values of $\rho=\pm 0.9$.  We also exclude a very small number of OTM cases where Monte Carlo prices are less than 1 cent.  While columns 1 to 3 of the table compare the alternative strike conventions of $a=0$ and $a=1$ with our optimal $a^\star$, the final column shows the `at-the-money (ATM) error', meaning the error averaged over only the cases where $S_0^Y=S_0^X=100$.  Recall from Figure \ref{TestCases} that ATM prices agree across all strike conventions (for any $a$) since they all collapse onto the same choice of $k_X,k_Y$.  As discussed in Section 2, ATM error is zero as $T\to 0$, but is non-zero here since $T\ge 0.05$.  In some sense, ATM error is thus the best we could hope for our strike convention to reach when averaging across all moneyness values.  \\

\begin{table}
 \begin{center}
  \begin{tabular}{|l|c||ccc|c|c|}
\hline
& $\rho$ & $a=0$ & $a=1$ & $a^\star$& bounded $a^\star$ & ATM error \\
\hline
\multirow{5}{*}{\rotatebox[origin=c]{90}{$T=0.05$}} 
&-0.7&	0.0646&	0.053&	0.0069&	0.0069&	0.0036\\
&-0.3&	0.0591&	0.0241&	0.0066&	0.0066&	0.0048\\
&0.1&	0.0567&	0.0064&	0.0107&	0.01&	0.005\\
&0.5&	0.0408&	0.0189&	0.0193&	0.0117&	0.0039\\
&0.9&	0.0147&	0.0121&	0.005&	0.0052&	0.0018\\
\hline
\multirow{5}{*}{\rotatebox[origin=c]{90}{$T=0.1$}} 
&-0.7&	0.1108&	0.0838&	0.0117&	0.0117&	0.01\\
&-0.3&	0.1064&	0.0395&	0.0132&	0.0132&	0.012\\
&0.1&	0.1093&	0.0159&	0.0217&	0.0203&	0.0121\\
&0.5&	0.091&	0.0439&	0.0426&	0.0277&	0.0102\\
&0.9&	0.0369&	0.0304&	0.0126&	0.0138&	0.004\\
\hline
\multirow{5}{*}{\rotatebox[origin=c]{90}{$T=0.25$}} 
&-0.7&	0.1943&	0.1413&	0.039&	0.039&	0.0426\\
&-0.3&	0.1932&	0.0787&	0.0465&	0.0465&	0.0487\\
&0.1&	0.2074&	0.0497&	0.0561&	0.0536&	0.0441\\
&0.5&	0.1933&	0.1004&	0.095&	0.0681&	0.0345\\
&0.9&	0.1145&	0.0948&	0.0434&	0.0465&	0.0181\\
\hline
\multirow{5}{*}{\rotatebox[origin=c]{90}{$T=1$}} 
&-0.7&	0.3634&	0.3097&	0.2243&	0.2243&	0.2252\\
&-0.3&	0.3976&	0.2855&	0.2632&	0.2632&	0.2724\\
&0.1&	0.4134&	0.2436&	0.2519&	0.2488&	0.2398\\
&0.5&	0.3898&	0.2575&	0.2464&	0.2195&	0.1856\\
&0.9&	0.3246&	0.2735&	0.1471&	0.156&	0.102\\
\hline
  \end{tabular}
 \end{center}
\caption{Comparison of strike conventions by Mean Absolute Error (MAE) averaged across $\rho_X$, $\rho_Y$ and $S^Y_0$ grids, varying $\rho$ and $T$ as labelled on left.}
\label{MAE1}
\end{table}

As we see in Table \ref{MAE1}, the optimal $a^\star$ outperforms the other strike conventions in the vast majority of cases, often cuts MAE by more than 50\% versus $a=0$ or $a=1$, and comes much closer to the ATM error.  Interestingly, $a=1$ is much more competitive than $a=0$ and seems to slightly outperform $a^\star$ as a convention when $\rho$ is near zero.  However this is not so surprising considering that $a^\star$ is often near 1 anyway in such cases, in line with the first special case in Remark \ref{specialcases} earlier.  Furthermore, the weakest cases of performance can often be attributed to unusually large (or very negative) values of $a^\star$, since they imply picking implied volatilities from deep ITM or OTM vanilla options, especially when $\left|S_0^X-S_0^Y\right|$ is not small.  This is of course also impractical in the real world.  As a possible improvement, in the final column of the table we show the average pricing errors for when bounding $a^\star$ in the range $[-1,2]$.  The extreme $a^\star$ situation is more common for cases of positive and fairly high $\rho$.  For example, for $\rho=0.5$ here, $a^\star$ happens to reach as high as 7.6 and as low as -3.7 at some gridpoints.  Therefore, while the bounding of $a^\star$ in $[-1,2]$ does not affect all rows, for $\rho=0.5$ it narrows the gap between $a^\star$ and ATM error by about 50\%.  Tests on data would be required to better assess the impact of this point, but we leave this for further studies.

\begin{figure}[htbp]
\centering
\includegraphics[width=0.85\textwidth]{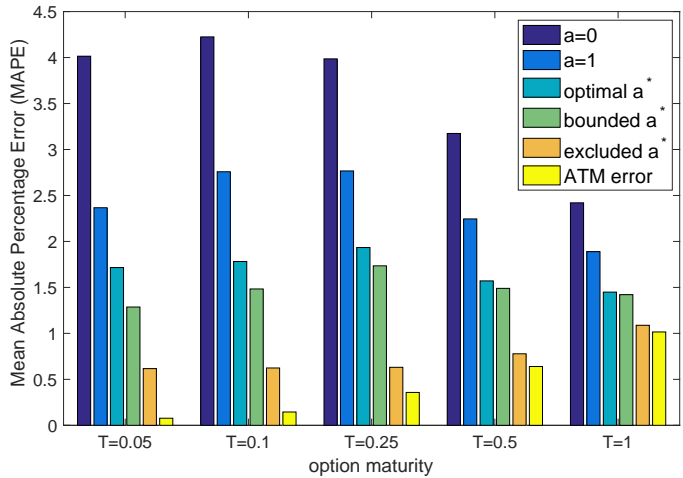}
\caption{Comparison of strike conventions by Mean Absolute Percentage Errors (MAPE) for various $T$ (incl. ATM error and bounding / excluding high $|a^\star|$)}
\label{barchart}
\end{figure}

In addition to our earlier parameter set with very short maturity $T=0.05$, we are also interested in investigating the performance of the approach for larger $T$.  Moving down Table \ref{MAE1}, results for longer maturity reveal that even without bounding (or excluding) trickier cases of high $|a^\star|$, our approach continues to perform well, always substantially outperforming $a=0$ and often significantly outperforming $a=1$ especially for higher $|\rho|$.  Interestingly, although the theory for $a^\star$ was derived for short time to maturity, we see that the approach maintains a competitive advantage for large $T$, even $T=1$.  Overall MAE levels are higher in all cases when $T$ increases, but the increase stems from option prices being higher and from ATM error increasing, while the gap between $a^\star$ and ATM error narrows to near zero.  Since larger $T$ clearly implies larger option prices, it is insightful here to also consider MAPE.  Figure \ref{barchart} reveals the average MAPE across all cases (including the 10 values of $\rho$) split by $T$ this time.  Seen in percentage terms, ATM error grows steadily with $T$, but error from all strike conventions actually falls.  Our strike convention $a^\star$ maintains a 0.5\%-1.0\% advantage over $a=1$ across maturities, and the bounded version improves this slightly.  Moreover, if we exclude the more challenging gridpoints with $a^\star \notin [-1,2]$, the plot shows that MAPE falls significantly to be very close to ATM error especially for larger $T$.  \\

\begin{table}
\hskip -1.5cm
  \begin{tabular}{|l|c|cccc|c|ccc|c|}
\hline
&& \multicolumn{5}{|c|}{No exclusions (normal case)}  &   \multicolumn{4}{|c|}{$a^\star \notin [-1,2]$ excluded} \\
\hline
& $\rho$ & $a=0$ & $a=1$& $a^\star$ & bounded & ATM &$a=0$ & $a=1$ & $a^\star$ & ATM  \\
\hline
\multirow{5}{*}{\rotatebox[origin=c]{90}{$T=0.1$}} 
&MAE&	0.0982&	0.0502&	0.022&	0.0177&	0.01&	0.0949&	0.0453&	0.013&	0.01\\
&MAPE&	4.23\%&	2.76\%&	1.78\%&	1.48\%&	0.14\%&	3.18\%&	1.89\%&	0.62\%&	0.15\%\\
&RMSE&	0.1288&	0.0649&	0.038&	0.0248&	0.0121&	0.1235&	0.0588&	0.016&	0.0122\\
&MaxAE&	0.3616&	0.1768&	0.2338&	0.1041&	0.0245&	0.3445&	0.1563&	0.0408&	0.0244\\
&MStd&	0.1137&	0.0567&	0.0169&	0.0149&	n/a&	0.1098&	0.051&	0.0098&	n/a\\

\hline
\multirow{5}{*}{\rotatebox[origin=c]{90}{$T=0.25$}} 
&MAE&	0.1908&	0.1029&	0.0586&	0.0512&	0.0381&	0.1806&	0.0908&	0.04&	0.0381\\
&MAPE&	3.99\%&	2.77\%&	1.93\%&	1.73\%&	0.36\%&	2.53\%&	1.51\%&	0.63\%&	0.37\%\\
&RMSE&	0.2499&	0.1307&	0.0897&	0.0688&	0.0459&	0.2356&	0.1154&	0.0481&	0.0455\\
&MaxAE&	0.7359&	0.3678&	0.4695&	0.274&	0.0914&	0.6899&	0.3145&	0.119&	0.0881\\
&MStd&	0.2188&	0.1067&	0.0365&	0.0336&	n/a&	0.2078&	0.0928&	0.0205&	n/a\\
\hline
\end{tabular}
\caption{Comparison of all results for five different error measures with all points included (left) and excluding $a^\star<-1,a^\star>2$ cases (right).  Results shown for all $T=0.05$ scenarios (top) and all $T=0.25$ scenarios (bottom).}
\label{allerrors}
\end{table}

We focused more on MAE above primarily due to the observation in Figure 1 that relative errors show a clear asymmetry between ITM and OTM which could distort strike convention comparisons in different cases.  However, Table \ref{allerrors} illustrates how our 1-STOSC approach compares to the other conventions across all our different error measures when averaging over all the scenarios for $T=0.1$ and $T=0.25$.   The left half of the table includes all cases of $a^\star$ (as in Table \ref{MAE1}), while the right half simply excludes cases where $a^\star<-1$ or $a^\star>2$, as mentioned above in Figure \ref{barchart}.  The fourth column also shows the middle-ground of a `bounded' $a^\star$ within this range instead of excluding these gridpoints.\footnote{Note that the $a=0$, $a=1$ and ATM columns also change slightly (often improve a little) when excluding these more extreme cases from the average error.} Throughout the table the optimal strike convention performs very well again, and depending on the error measure used, bounding $a^\star$ can cut the gap to ATM error in half, while exclusions may bring us almost all the way.  However, what is especially crucial is the clear benefit $a^\star$ already provides relative to a commonly-used choice such as ATM implied vols ($a=0$), often reducing error by a factor of about 3 or 4.  
 
\begin{figure}[htbp]
\centering
\includegraphics[width=0.95\textwidth]{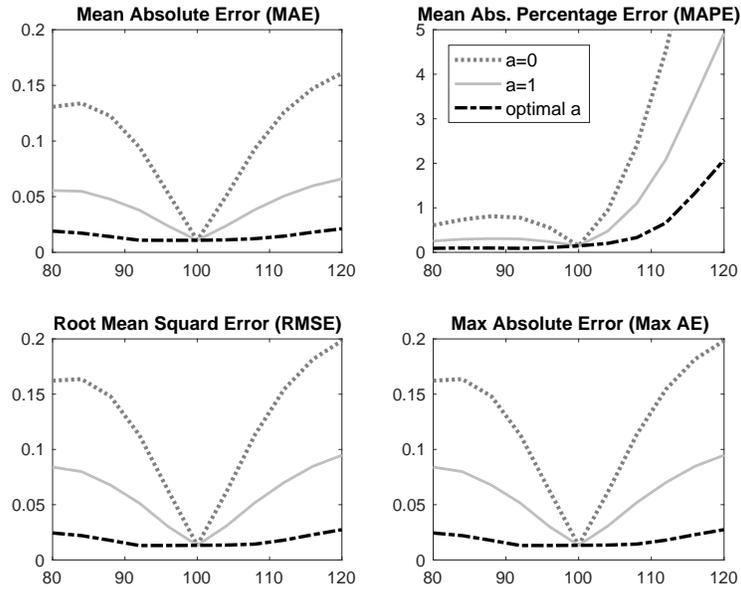}
\caption{Comparison of average errors against moneyness for all $T=0.1$ cases, using four different error measures}
\label{moneynessplot}
\end{figure}

Finally, before concluding we return to the question of consistency across moneyness, a key strength of the approach which is captured well by the impressive final row of the table called `MStd' (maximum standard deviation), but is also visually striking in Figure \ref{moneynessplot}.  Here we plot average errors across moneyness (against $S_0^Y$ again) average over all the $T=0.1$ grids.  Backing up the theory derived in earlier sections, the stability of errors across moneyness is very prominent, especially in comparison with $a=0$ or $a=1$, the commonly-used alternatives.  Indeed, to our knowledge there is no other approach which adapts the strike convention to different scenarios in order to achieve such clear-cut error reduction.  
 
\section{Conclusion}
We have presented a new and systematic methodology to construct an optimal strike convention for spread option pricing in the context of stochastic volatility models. Although its derivation is rather technical, this approach is simple to use and is based on the computation of the corresponding vanilla implied volatility levels and skews.  Thus, market observables can be taken as inputs in a model-independent setting, strengthening the appeal of the technique.  The obtained numerical results in Section 5 confirm its strong performance, especially compared to the limited alternatives commonly used in industry.  There is more interesting work to be done in this direction, for example extending from exchange options to any spread options or to three-asset spreads.  Data analysis and further numerical investigations would also be useful, including adapting to other stochastic volatility processes such as fractional models.  We thus see this paper as the starting point to a broadly-applicable and valuable new pricing tool designed to complement nicely existing practice in the financial markets.

\end{document}